\documentclass[11pt]{article}

\usepackage[margin=1in]{geometry}

\usepackage{amsmath}
\usepackage{amssymb}
\usepackage{amsthm}
\usepackage{mathtools}

\usepackage{booktabs}
\usepackage{array}
\usepackage{multirow}
\usepackage{tabularx}

\usepackage{listings}
\usepackage{xcolor}

\usepackage[colorlinks=true,linkcolor=blue,citecolor=blue,urlcolor=blue]{hyperref}

\usepackage{natbib}

\usepackage{tikz}
\usetikzlibrary{positioning, arrows.meta, fit, backgrounds, calc, shapes.geometric}

\usepackage{enumitem}
\usepackage{microtype}
\usepackage{graphicx}

\newtheorem{theorem}{Theorem}[section]
\newtheorem{proposition}[theorem]{Proposition}

\theoremstyle{definition}
\newtheorem{definition}[theorem]{Definition}
\theoremstyle{remark}
\newtheorem{remark}[theorem]{Remark}

\definecolor{codebg}{rgb}{0.95,0.95,0.95}
\definecolor{codegreen}{rgb}{0,0.6,0}
\definecolor{codepurple}{rgb}{0.58,0,0.82}
\definecolor{codeblue}{rgb}{0.0,0.0,0.7}

\lstdefinestyle{formal}{
  backgroundcolor=\color{codebg},
  basicstyle=\ttfamily\small,
  breaklines=true,
  frame=single,
  framerule=0pt,
  xleftmargin=1em,
  xrightmargin=1em,
  aboveskip=0.5em,
  belowskip=0.5em,
  mathescape=true,
  showstringspaces=false,
}

\lstdefinestyle{json}{
  backgroundcolor=\color{codebg},
  basicstyle=\ttfamily\small,
  breaklines=true,
  frame=single,
  framerule=0pt,
  xleftmargin=1em,
  xrightmargin=1em,
  aboveskip=0.5em,
  belowskip=0.5em,
  showstringspaces=false,
}

\lstdefinestyle{toml}{
  backgroundcolor=\color{codebg},
  basicstyle=\ttfamily\small,
  breaklines=true,
  frame=single,
  framerule=0pt,
  xleftmargin=1em,
  xrightmargin=1em,
  aboveskip=0.5em,
  belowskip=0.5em,
  showstringspaces=false,
  commentstyle=\color{gray}\itshape,
  morecomment=[l]{\#},
}

\title{Cryptographic Registry Provenance:\\Structural Defense Against Dependency Confusion\\in AI Package Ecosystems}

\author{Alan L. McCann\\
  \textit{Mashin, Inc.}\\
  \texttt{research@mashin.live}
}

\date{April 2026}

\begin{document}

\maketitle

\begin{abstract}
Dependency confusion attacks exploit a structural gap in software distribution: once a package is installed, there is no cryptographic proof of which registry distributed it. Every existing defense is configuration-based and fails silently when misconfigured. We present a cryptographic distribution provenance system comprising three components: (1)~\emph{cryptographic registry identity}, where every registry holds an Ed25519 keypair and signs every artifact it distributes; (2)~a \emph{dual-signature model}, where the publisher signs at packaging time and the registry countersigns at publication time; and (3)~\emph{authoritative namespace binding}, where consumers pin registry fingerprints and the resolver cryptographically rejects artifacts from unauthorized registries. These create three defense layers requiring simultaneous compromise for a successful attack. A comparison across eight ecosystems (npm, Cargo, Hex.pm, PyPI, Go modules, Docker/OCI, NuGet, Maven) shows no existing ecosystem combines mandatory publisher signing, cryptographic registry identity, mandatory registry countersigning, and consumer-side cryptographic enforcement. The system extends to AI-generation provenance as a signed attribute and governance-enforced dependency resolution. A case study integrates distribution provenance with a three-layer runtime governance architecture, creating a four-phase lifecycle chain with no cryptographic gaps.
\end{abstract}

\section{Introduction}
\label{sec:introduction}

Software supply chain attacks have emerged as one of the most consequential threat categories in modern software engineering. The 2021 ua-parser-js incident affected over 7 million weekly npm downloads when an attacker gained access to the maintainer's account and published malicious versions~\cite{uaparserjs2021}. The 2018 event-stream attack inserted a targeted cryptocurrency theft module into a package with 1.5 million weekly downloads~\cite{eventstream2018}. The 2024 xz-utils backdoor demonstrated that even cryptographically signed packages are vulnerable when the signing key holder is socially engineered~\cite{xzutils2024}.

Among supply chain attack vectors, \emph{dependency confusion} occupies a distinctive position: it is structurally simple, reliably exploitable, and undefended at the artifact level in every major package ecosystem. First systematized by Birsan~\cite{birsan2021dependency}, the attack exploits a fundamental architectural gap: once a package is installed on a consumer's system, there is no cryptographic proof of which registry distributed it.

\subsection{The Structural Gap}

Consider an enterprise running a private package registry at \texttt{registry.acme.com} alongside the public registry. Both can host packages under the \texttt{@acme} namespace: the private registry legitimately, the public registry through namespace squatting or misconfiguration. The enterprise's build system resolves \texttt{@acme/utils} by querying configured registries. If the public registry returns a higher version number, the resolver may select it over the private registry's version.

Every existing mitigation operates at the configuration layer:

\begin{itemize}
  \item \textbf{npm}: \texttt{.npmrc} \texttt{registry} and \texttt{@scope:registry} settings
  \item \textbf{PyPI/pip}: \texttt{--index-url} and \texttt{--extra-index-url} flags
  \item \textbf{Go modules}: \texttt{GOPROXY} environment variable
  \item \textbf{Maven}: \texttt{<repositories>} XML configuration in \texttt{pom.xml}
  \item \textbf{NuGet}: \texttt{nuget.config} package source mappings (added 2022)
\end{itemize}

If the configuration is incorrect, missing, or overridden by a CI/CD environment variable, the attack succeeds silently. The artifact on disk carries no evidence that would allow post-installation verification of its distribution path.

\subsection{Contributions}

This paper makes the following contributions:

\begin{enumerate}
  \item \textbf{Cryptographic registry identity} (Section~\ref{sec:registry-identity}): every registry instance holds an Ed25519 keypair and publishes its identity at a well-known endpoint. The registry signs every artifact it accepts, creating an artifact-level cryptographic proof of the distribution path.

  \item \textbf{A dual-signature distribution model} (Section~\ref{sec:dual-signature}): published artifacts carry two independent signatures (a publisher signature created at packaging time and a registry attestation created at publication time) with countersigning semantics that preserve temporal ordering and independent claims.

  \item \textbf{Authoritative namespace binding} (Section~\ref{sec:namespace-binding}): consumers declare which registries are authoritative for which namespaces and pin registry key fingerprints. The package resolver cryptographically rejects artifacts from unauthorized registries.

  \item \textbf{A two-layer archive format} (Section~\ref{sec:archive-format}): a provenance envelope (uncompressed) wrapping source contents (compressed), enabling verification without full decompression.

  \item \textbf{A six-level verification chain} (Section~\ref{sec:verification-chain}): progressive verification from file integrity through publisher authenticity, registry attestation, and lineage provenance.

  \item \textbf{A systematic ecosystem comparison} (Section~\ref{sec:comparison}): analysis of eight major ecosystems showing no existing system combines all four capabilities (mandatory signing, registry identity, registry countersigning, consumer enforcement).

  \item \textbf{Case study: integration with runtime governance} (Section~\ref{sec:integration}): we demonstrate integration with a three-layer structural governance architecture~\cite{mccann2026structural}, extending it to a four-phase lifecycle chain with no cryptographic gaps.

\end{enumerate}

\noindent Two extensions (AI-generation provenance and governance-enforced dependency resolution) are discussed as future work directions in Section~\ref{sec:discussion}.

\paragraph{Scope and companion papers.}
This paper addresses the distribution phase of the governance lifecycle: proving which registry distributed an artifact. It completes a four-phase chain from build provenance through distribution to runtime governance.
\cite{mccann2026structural}~establishes the structural governance architecture that this paper extends to the supply chain.
\cite{mccann2026mechanized}~mechanizes the runtime safety properties (454 theorems in Rocq) that govern artifacts after installation.
\cite{mccann2026gcc}~proves that the four governed primitives are expressively complete and semantically transparent; registry verification operates at the syntactic level (cryptographic properties), where decidability holds.
\cite{mccann2026algebraic}~lifts these results to a parametric algebraic framework; the coterminous boundary principle applies equally to distribution verification: every package either passes full verification or is rejected.
\cite{mccann2026purity}~addresses the executor-level gap: purity certificates issued at build time compose with the distribution signatures described here to form a complete trust chain from source to execution.

\subsection{Paper Structure}

Section~\ref{sec:related} surveys related work in package security and supply chain provenance. Section~\ref{sec:threat-model} defines the threat model. Sections~\ref{sec:archive-format}--\ref{sec:verification-chain} present the architecture. Section~\ref{sec:formal} formalizes security properties. Section~\ref{sec:comparison} compares against existing ecosystems. Section~\ref{sec:integration} describes integration with runtime governance. Section~\ref{sec:evaluation} evaluates performance and security. Section~\ref{sec:discussion} discusses limitations, future work, and preliminary extensions. Section~\ref{sec:conclusion} concludes.

\section{Related Work}
\label{sec:related}

\subsection{Package Signing}

Package signing has evolved from optional to increasingly mandatory across ecosystems. npm introduced Sigstore-based provenance attestations in 2023~\cite{npm-provenance2023}, providing keyless signing tied to CI/CD identities via OIDC federation. PyPI adopted PEP~740 attestations in 2024~\cite{pep740}, also Sigstore-based. NuGet has supported X.509 author signatures since 2018 and repository countersignatures since 2019~\cite{nuget-signing}. Hex.pm signs packages with a per-repository key inside a protobuf envelope~\cite{hex-registry}. Cargo does not natively sign packages but the third-party \texttt{cargo-vet} tool provides supply chain auditing~\cite{cargo-vet}.

These efforts focus on \emph{publisher identity}, proving who built the artifact. None except NuGet and Hex address \emph{registry identity}: proving which registry distributed it. Even NuGet's repository countersignature is optional and not enforced during package installation by default~\cite{nuget-signing}.

\subsection{Sigstore and Keyless Signing}

Sigstore~\cite{sigstore2022} eliminated long-lived signing keys entirely. Publishers authenticate via OIDC (e.g., GitHub Actions identity), receive a short-lived certificate from Fulcio, sign the artifact, and the signature is recorded in the Rekor transparency log. Consumers verify by checking the transparency log rather than managing a trusted key set.

Sigstore excels at proving publisher identity in CI/CD-centric workflows. It does not address registry provenance: the transparency log records \emph{who signed}, not \emph{which registry distributed}. An artifact signed via Sigstore in a GitHub Actions workflow and uploaded to two different registries produces identical signatures; the distribution path is invisible.

\subsection{Software Bill of Materials and SLSA}

SLSA (Supply-chain Levels for Software Artifacts)~\cite{slsa2021} defines a maturity framework for supply chain security, from basic build provenance (Level~1) to hermetic builds with two-party review (Level~4). SLSA provenance attestations record the build environment and source repository but do not address registry-level provenance; \emph{how the artifact moved from build to consumer} is outside SLSA's scope.

SBOM standards (SPDX~\cite{spdx2021}, CycloneDX~\cite{cyclonedx2022}) inventory components within artifacts but do not provide cryptographic proof of distribution path. An SBOM answers ``what's in this artifact?'' but not ``which registry distributed this artifact?''

\subsection{Dependency Confusion Attacks}

Birsan~\cite{birsan2021dependency} systematized the dependency confusion attack in February 2021, demonstrating successful exploitation against Apple, Microsoft, PayPal, and 30+ other organizations. The attack exploits the resolution behavior of package managers that query multiple registries: if the public registry hosts a package with the same name as a private package but a higher version number, the resolver selects the public (attacker-controlled) package.

Subsequent defenses have been configuration-based: npm added scope-to-registry mappings in \texttt{.npmrc}~\cite{npm-scope-registry}, NuGet added package source mappings~\cite{nuget-source-mapping}, and organizations adopted internal proxying (Artifactory, Nexus) to control resolution. All defenses operate at the resolver configuration layer, not at the artifact level.

\subsection{Transparency Logs}

Certificate Transparency~\cite{ct-rfc6962} pioneered the use of append-only Merkle tree logs for detecting misissued TLS certificates. The Go module ecosystem's \texttt{sum.golang.org}~\cite{go-sumdb} applies this concept to package checksums: a transparency log records the expected checksum of every module version, and clients verify downloaded modules against the log. This prevents \emph{modification} of existing versions (a published module cannot be silently replaced) but does not address \emph{which registry served the module} or \emph{dependency confusion} (the log records all modules, including attacker-published ones).

Sigstore's Rekor transparency log~\cite{sigstore2022} records signing events. As with \texttt{sum.golang.org}, the log provides \emph{append-only evidence} but not \emph{registry-level provenance}.

\subsection{Secure Update and Distribution Frameworks}

\textbf{The Update Framework} (TUF)~\cite{samuel2010survivable} provides a framework for secure software updates that survives key compromise through threshold signing, role separation, and key rotation. TUF's threat model assumes a single trusted repository and focuses on update integrity: ensuring that downloaded updates are authentic and current. Our model addresses a complementary problem: when multiple registries exist, proving which registry distributed an artifact. TUF provides no mechanism for registry identity or namespace authority binding. The two frameworks could be composed: TUF securing updates within a single Kura instance, while our provenance chain secures the cross-registry distribution path.

\textbf{in-toto}~\cite{torres2019intoto} extends TUF to verify that each step in a software supply chain was performed by an authorized party. A supply chain \emph{layout} specifies the expected sequence of steps (build, test, package, sign); in-toto verifies that link metadata from each step matches the layout. in-toto addresses \emph{build provenance}: was each build step performed correctly? Our architecture addresses \emph{distribution provenance}: which registry published this artifact, and does it have namespace authority? in-toto could serve upstream of our provenance chain: an artifact's in-toto attestation proves how it was built, while our registry provenance proves how it was distributed.

\paragraph{Layered complementarity.} TUF and in-toto each operate at a different layer from the present work, and the three are complementary rather than competing. TUF ensures distribution integrity within a single repository: that the registry serves what it intends to serve, and that key compromise does not silently corrupt updates. The present system's dual-signature model addresses a distinct question: which \emph{specific} registry distributed a given artifact, and does that registry have namespace authority? in-toto verifies build process integrity, confirming that each step in the build pipeline was performed by an authorized party. The present system verifies runtime governance compliance, confirming that an artifact meets trust tier, effect permission, and composition level requirements at install time. The approaches compose naturally: in-toto attestations can prove how an artifact was built, TUF can secure updates within a single registry instance, and our provenance chain can prove the cross-registry distribution path. We adopt in-toto's attestation format~\cite{torres2019intoto} for evidence bundles in the provenance manifest, maintaining interoperability with existing supply chain tooling.

\textbf{Notary v2}~\cite{notaryv2} provides OCI artifact signing for container registries, enabling publishers to sign container images and consumers to verify signatures at pull time. Notary v2 addresses registry-level signing for containers specifically. Our model is registry-agnostic (not limited to OCI artifacts) and includes namespace authority binding, which Notary v2 does not provide. The dual-signature model (publisher signs, registry countersigns) is conceptually similar to NuGet's repository countersignature but applied at the distribution provenance level rather than the package metadata level.

Linux distribution package managers have decades of experience with registry-level signing. \textbf{Debian's Release signing}~\cite{debian2023secure} uses GPG to sign Release files, providing a chain from the archive key to individual package hashes. \textbf{RPM's GPG signing}~\cite{rpm2023signing} signs individual packages, with consumers configuring trusted GPG keys per repository. Both models implement a form of registry identity: the signing key identifies the distribution. Our model extends this concept with explicit namespace authority binding (which registry is authoritative for which namespaces), dual signatures (publisher and registry), and hash-chain provenance that links distribution to execution.

\subsection{Attack Taxonomies}

Ohm et al.~\cite{ohm2020backstabber} provide a systematic taxonomy of attacks on package repositories, cataloging attack vectors including typosquatting, dependency confusion, maintainer compromise, and build system exploitation. Their ``Backstabber's Knife Collection'' dataset enables quantitative analysis of attack prevalence. Our architecture addresses three of the four primary attack vectors: dependency confusion (through authoritative namespace binding), maintainer account compromise (through TOFU key pinning), and malicious package injection (through dual-signature verification). Typosquatting remains outside our scope, as it is a naming-layer attack rather than a distribution-layer attack.

Ladisa et al.~\cite{ladisa2023taxonomy} present the most comprehensive academic taxonomy of open-source supply chain attacks, organizing attacks along the supply chain lifecycle. Their framework identifies the distribution phase as a distinct attack surface. Our contribution provides structural defenses specifically for this phase, mapping to their taxonomy's ``Distribute'' category. The four-phase lifecycle (Section~\ref{sec:integration}) was designed with this attack surface taxonomy in mind.

\section{Threat Model}
\label{sec:threat-model}

We consider the following adversary capabilities:

\begin{definition}[Threat Model]
An adversary can:
\begin{enumerate}[label=(\alph*)]
  \item Publish artifacts to any public registry under any namespace not protected by organizational membership verification.
  \item Observe the names and versions of private packages used within a target organization (via leaked lock files, CI logs, error messages, or DNS queries to private registries).
  \item Compromise CI/CD environment variables, \texttt{.npmrc} files, or equivalent resolver configuration.
  \item Operate a malicious registry that mimics a legitimate registry's API.
\end{enumerate}
An adversary \emph{cannot}:
\begin{enumerate}[label=(\alph*),resume]
  \item Compromise the target organization's private registry server.
  \item Obtain the private Ed25519 keys of legitimate publishers or registries.
  \item Modify artifacts in transit (TLS is assumed for all registry communication).
\end{enumerate}
\end{definition}

\noindent This threat model covers dependency confusion (capabilities a--c), namespace squatting (a), and resolver hijacking (c--d). It does not cover insider threats with registry server access (e) or key compromise (f), which require orthogonal defenses (access control, key rotation, hardware security modules).

We define a successful attack as: the consumer's build system installs an artifact that was not published by an authorized publisher to an authorized registry for the artifact's namespace, without the build system detecting the substitution.

\section{Two-Layer Archive Format}
\label{sec:archive-format}

Published artifacts use a two-layer archive format separating the provenance envelope from the source contents:

\begin{lstlisting}[style=formal]
artifact-1.2.0.pkg (uncompressed outer tar)
$\vdash$ provenance.json            # Provenance manifest
$\vdash$ signature.json            # Publisher Ed25519 signature
$\vdash$ registry_attestation.json # Registry countersignature
$\vdash$ CHECKSUM                  # SHA-256 of contents.tar.gz
$\vdash$ contents.tar.gz           # Gzipped tar of source files
    $\vdash$ manifest.toml          # Project manifest
    $\vdash$ src/                   # Source files
    $\vdash$ types/                 # Type definitions
    $\vdash$ README.md
\end{lstlisting}

\subsection{Design Rationale}

The outer envelope is an \emph{uncompressed} tar archive. The inner contents are a \emph{gzipped} tar archive. This separation provides two properties:

\paragraph{Fast verification.} The registry reads the provenance manifest, verifies signatures, and generates its attestation without decompressing the source archive. On install, the consumer verifies the provenance envelope before extracting the inner archive. Verification cost is proportional to the small envelope ($\sim$2--5~KB), not the potentially large source archive.

\paragraph{Atomic provenance.} The provenance documents (provenance manifest, signatures, attestation, checksum) travel with the artifact in a single file. There is no separate ``provenance sidecar'' that could become detached from the artifact it describes. The artifact \emph{is} its provenance.

This design follows Hex.pm's two-layer approach~\cite{hex-registry}, extending it with the registry attestation document and the machine-generated provenance manifest. Figure~\ref{fig:archive} illustrates the complete archive structure and dual-signature flow.

\subsection{The Provenance manifest}

The provenance manifest (\texttt{provenance.json}) is a machine-generated metadata document computed during the packaging process. It contains:

\begin{itemize}
  \item \textbf{Artifact identity}: name, version, content hash (SHA-256 of canonical serialization of all source files), package identifier, namespace.
  \item \textbf{Module-level metadata}: for each component in the artifact (name, qualified name, description, input/output schemas, per-file SHA-256 checksum).
  \item \textbf{Dependencies}: declared dependency constraints.
  \item \textbf{Lineage}: parent version and content hash (creating a version chain), fork origin (if forked from another artifact), and an optional \emph{evolution anchor}: a reference to a specific event in an external append-only provenance ledger~\cite{mccann2026structural} that produced this version.
  \item \textbf{Build metadata}: timestamp, runtime version.
\end{itemize}

The provenance manifest bridges the human-authored project manifest (inside the compressed inner archive) and the cryptographically verifiable identity (in the uncompressed outer envelope).

\begin{definition}[Content Hash]
\label{def:content-hash}
The content hash of an artifact is:
\[
H_c = \text{SHA-256}(\text{canonical}(\textit{name} \| \textit{version} \| \textit{pkg\_id} \| \textit{sorted\_sources}))
\]
where $\text{canonical}(\cdot)$ produces a deterministic byte representation with sorted keys and normalized values. Semantically identical artifacts produce identical content hashes regardless of file system ordering or equivalent value representations.
\end{definition}

\begin{figure*}[t]
\centering
\begin{tikzpicture}[
  scale=0.92, every node/.style={transform shape},
  >=Stealth,
  box/.style={draw, rounded corners=2pt, minimum height=1.6em, minimum width=2.4cm, font=\footnotesize\ttfamily, fill=codebg},
  phase/.style={draw, rounded corners=4pt, minimum height=2em, minimum width=2.8cm, font=\small\bfseries, fill=blue!8},
  sig/.style={draw, rounded corners=2pt, minimum height=1.5em, font=\footnotesize\ttfamily, fill=yellow!15},
  arrow/.style={->, thick},
]


\node[draw, thick, rounded corners=6pt, minimum width=11.5cm, minimum height=4.8cm, fill=white, label={[font=\small\bfseries]above:Published Artifact (\texttt{artifact-1.2.0.pkg})}] (outer) at (0,0.8) {};

\node[font=\scriptsize\itshape, anchor=north east] at ($(outer.north east)+(-0.15,-0.15)$) {uncompressed outer tar};

\node[box] (prov) at (-3.5, 1.9) {provenance.json};
\node[sig] (sig) at (-3.5, 1.0) {signature.json};
\node[sig, fill=green!12] (att) at (-3.5, 0.1) {registry\_attestation.json};
\node[box] (check) at (-3.5, -0.8) {CHECKSUM};

\node[draw, thick, rounded corners=4pt, minimum width=4.5cm, minimum height=3.8cm, fill=blue!3, anchor=west] (inner) at (-0.3, 0.55) {};
\node[font=\footnotesize\bfseries, anchor=north] at ($(inner.north)+(0,-0.1)$) {contents.tar.gz};
\node[font=\scriptsize\itshape, anchor=north] at ($(inner.north)+(0,-0.45)$) {(compressed)};

\node[box, minimum width=3.2cm, fill=white] (manifest) at ($(inner.center)+(0, 0.4)$) {manifest.toml};
\node[box, minimum width=3.2cm, fill=white] (src) at ($(inner.center)+(0, -0.35)$) {src/};
\node[box, minimum width=3.2cm, fill=white] (types) at ($(inner.center)+(0, -1.05)$) {types/, README};


\node[phase, fill=blue!12, minimum width=2.2cm] (pub) at (-6, -2.8) {Publisher};
\node[phase, fill=green!10, minimum width=2.2cm] (reg) at (0, -2.8) {Registry};
\node[phase, fill=red!8, minimum width=2.2cm] (con) at (6, -2.8) {Consumer};

\draw[arrow, thick] (pub) -- node[above, font=\scriptsize, yshift=4pt] {$t_1$: package \& sign} (reg);
\draw[arrow, thick] (reg) -- node[above, font=\scriptsize, yshift=4pt] {$t_2$: verify \& countersign} (con);

\node[font=\scriptsize, align=center, anchor=north] at (pub.south) {packages artifact,\\signs provenance manifest};
\node[font=\scriptsize, align=center, anchor=north] at (reg.south) {verifies $\sigma_P$, adds\\registry attestation $\sigma_R$};
\node[font=\scriptsize, align=center, anchor=north] at (con.south) {verifies $\sigma_P$, $\sigma_R$,\\checks $F_R = F_{\mathit{pinned}}$};

\end{tikzpicture}
\caption{Two-layer archive format with dual-signature flow. The uncompressed outer tar contains the provenance envelope (provenance manifest, publisher signature, registry attestation, checksum) wrapping the compressed inner archive. The publisher signs at packaging time ($t_1$), the registry countersigns at publication time ($t_2 > t_1$), and the consumer verifies both signatures on install.}
\label{fig:archive}
\end{figure*}
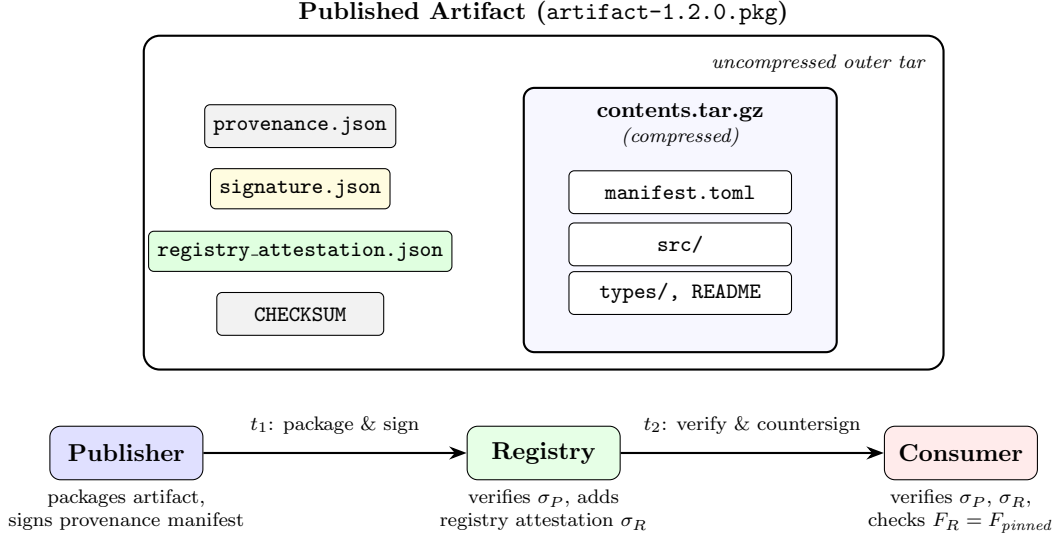

\section{Cryptographic Registry Identity}
\label{sec:registry-identity}

\begin{definition}[Registry Identity]
\label{def:registry-identity}
A registry identity is a tuple $R = (r, u, K_R, N_R, p)$ where:
\begin{itemize}
  \item $r$ is a unique registry identifier
  \item $u$ is the registry URL
  \item $K_R = (k_R^{\text{pub}}, k_R^{\text{priv}})$ is an Ed25519 keypair
  \item $N_R \subseteq \mathcal{N}$ is the set of namespaces the registry claims authority over
  \item $p$ is an optional parent registry identifier (for federated resolution)
\end{itemize}
\end{definition}

Every registry instance (public, private, or organization-scoped) generates an Ed25519 keypair on initialization. The private key is stored on the registry server. The public key and identity metadata are published at a well-known endpoint:

\begin{lstlisting}[style=json]
GET /.well-known/package-registry.json

{
  "registry_id": "reg_01JMXYZ...",
  "registry_url": "https://registry.acme.com",
  "public_key": "ed25519:<base64 public key>",
  "key_fingerprint": "sha256:<hex of SHA-256(public key)>",
  "namespaces": ["@acme", "@acme-internal"],
  "parent_registry": "reg_01JPUBLIC...",
  "key_valid_from": "2026-01-15T00:00:00Z",
  "key_rotation_history": []
}
\end{lstlisting}

\paragraph{Key fingerprint.} The key fingerprint $F_R = \text{SHA-256}(k_R^{\text{pub}})$ provides a short, human-verifiable identifier for the registry's identity. Consumers pin this fingerprint in their project configuration (Section~\ref{sec:namespace-binding}).

\paragraph{Namespace claims.} Each registry declares which namespaces it is authoritative for. A private registry at \texttt{registry.acme.com} declares authority over \texttt{@acme}. The public registry declares authority over \texttt{@stdlib} and \texttt{@community}. Namespace claims are part of the published identity and are verifiable by consumers.

\paragraph{Trust hierarchy.} Private registries can declare a parent registry, creating a tree: public $\to$ organization $\to$ team. When a private registry receives a query for a namespace outside its claims, it can proxy to its parent. The parent relationship is declared in the registry's published identity.

\paragraph{Key rotation.} The registry identity document includes a \texttt{key\_rotation\_history} array recording previous keys with validity periods. Consumers verifying a registry attestation check the \texttt{accepted\_at} timestamp against key validity periods to determine which key was active at attestation time. This enables key rotation without invalidating previously published artifacts.

\section{Dual-Signature Distribution Model}
\label{sec:dual-signature}

Every published artifact carries two independent cryptographic signatures.

\begin{definition}[Publisher Signature]
\label{def:publisher-sig}
During packaging, the publisher computes:
\[
\sigma_P = \text{Ed25519-Sign}(k_P^{\text{priv}},\; \text{canonical}(B))
\]
where $k_P^{\text{priv}}$ is the publisher's Ed25519 private key and $B$ is the provenance manifest. The signature document \texttt{signature.json} contains $\sigma_P$, the publisher's key fingerprint $F_P = \text{SHA-256}(k_P^{\text{pub}})$, and the signing timestamp.
\end{definition}

\begin{definition}[Registry Attestation]
\label{def:registry-attestation}
Upon accepting an artifact for publication, the registry:
\begin{enumerate}
  \item Verifies $\sigma_P$ against the publisher's registered public key.
  \item Verifies $H_c$ by recomputing from extracted source files.
  \item Verifies the publisher is authorized for the artifact's namespace.
  \item Optionally performs security scanning.
  \item Generates an attestation document $A$ containing:
    \begin{itemize}
      \item Registry identity: $r$, $u$, $F_R$
      \item Provenance manifest hash: $\text{SHA-256}(\text{canonical}(B))$
      \item Namespace, artifact name, version
      \item Acceptance timestamp
      \item Verification results (which checks were performed and passed)
    \end{itemize}
  \item Computes: $\sigma_R = \text{Ed25519-Sign}(k_R^{\text{priv}},\; \text{canonical}(A))$
\end{enumerate}
The registry injects \texttt{registry\_attestation.json} (containing $A$ and $\sigma_R$) into the outer envelope before storing the artifact.
\end{definition}

\subsection{Why Countersigning}

The dual-signature model uses \emph{countersigning} semantics: the publisher signs first, then the registry countersigns. This is architecturally superior to co-signing for three reasons:

\paragraph{Temporal ordering.} The publisher's signature proves they built the artifact at time $t_1$. The registry's attestation proves it accepted the artifact at time $t_2 > t_1$. These are temporally distinct events. Co-signing would conflate them into a single timestamp.

\paragraph{Independent claims.} The publisher asserts: ``I built this artifact with these contents.'' The registry asserts: ``I received this artifact, verified it, and am distributing it.'' These are logically independent; the publisher's claim is about authorship, the registry's about distribution. Independent assertions should carry independent signatures.

\paragraph{Offline packaging.} A developer can package and sign an artifact without network access. The registry attestation is added when they publish (online). Co-signing would require both parties online simultaneously.

\subsection{Comparison with NuGet}

NuGet~\cite{nuget-signing} is the closest prior art. NuGet supports author signatures (X.509) and repository countersignatures (X.509). The present system differs in three ways:

\begin{enumerate}
  \item \textbf{Mandatory vs.\ optional.} NuGet's repository countersignature is optional. In our system, every published artifact must carry a registry attestation.
  \item \textbf{Ed25519 vs.\ X.509.} Ed25519 keys are 32 bytes, have no certificate chain, and require no PKI infrastructure. X.509 certificates require certificate authorities, chain validation, and revocation checking (CRL/OCSP).
  \item \textbf{Enforced on install.} NuGet does not verify repository countersignatures during package installation by default. Our system verifies the registry attestation on every install and, when authoritative namespace binding is configured, rejects artifacts from unauthorized registries.
\end{enumerate}

\section{Authoritative Namespace Binding}
\label{sec:namespace-binding}

The dual-signature model provides cryptographic \emph{evidence} of the distribution path. Authoritative namespace binding provides cryptographic \emph{enforcement}: the consumer's system rejects artifacts not distributed through the expected registry.

\subsection{Consumer Configuration}

Consumers declare trusted registries and their namespace authority in the project manifest:

\begin{lstlisting}[style=toml]
# project manifest (e.g., manifest.toml)
[registries.acme]
url = "https://registry.acme.com"
fingerprint = "sha256:a1b2c3d4e5f6..."
namespaces = ["@acme", "@acme-internal"]
priority = "authoritative"

[registries.public]
url = "https://registry.example.com"
namespaces = ["@stdlib", "@community"]
\end{lstlisting}

\begin{definition}[Authoritative Binding]
\label{def:authoritative}
A namespace $n$ is \emph{authoritatively bound} to registry $R$ in project $P$ if $P$'s manifest declares $R$ with priority \texttt{authoritative} and $n \in N_R$ where $N_R$ is $R$'s declared namespace set. When $n$ is authoritatively bound to $R$:
\begin{enumerate}
  \item The resolver queries only $R$ for artifacts in namespace $n$.
  \item Artifacts in namespace $n$ found on any other registry are rejected.
  \item The registry attestation's fingerprint $F_R$ must match the pinned fingerprint in $P$'s manifest.
\end{enumerate}
\end{definition}

\subsection{TOFU Key Pinning}

The fingerprint field implements Trust On First Use (TOFU), analogous to SSH's \texttt{known\_hosts}~\cite{ssh-rfc4251}. On first configuration, the consumer pins the registry's key fingerprint. Subsequent key changes (from rotation, compromise, or man-in-the-middle) cause verification failure until the consumer explicitly updates the pinned fingerprint.

\begin{definition}[TOFU Verification]
\label{def:tofu}
Given a registry attestation with fingerprint $F_A$ and a consumer's pinned fingerprint $F_{\mathit{pinned}}$:
\[
\text{verify\_registry}(F_A, F_{\mathit{pinned}}) =
\begin{cases}
\text{accept} & \text{if } F_A = F_{\mathit{pinned}} \\
\text{reject} & \text{if } F_A \neq F_{\mathit{pinned}}
\end{cases}
\]
\noindent (We write $F_{\mathit{pinned}}$ to distinguish the consumer's pinned registry fingerprint from $F_P$, the publisher's key fingerprint defined in Definition~\ref{def:publisher-sig}.)
\end{definition}

\subsection{Resolution with Enforcement}

When the package resolver processes a dependency on artifact $a$ in namespace $n$:

\begin{enumerate}
  \item Consult the project manifest's \texttt{[registries]} configuration for namespace $n$.
  \item If $n$ is authoritatively bound to registry $R$: query only $R$. If $n$ has no explicit binding: query the default registry.
  \item Download the artifact. Read the outer envelope without decompressing the inner archive.
  \item Verify the publisher signature: $\text{Ed25519-Verify}(k_P^{\text{pub}}, \text{canonical}(B), \sigma_P)$.
  \item Verify the registry attestation:
    \begin{enumerate}
      \item $\text{Ed25519-Verify}(k_R^{\text{pub}}, \text{canonical}(A), \sigma_R)$
      \item If $n$ is authoritatively bound: $F_R = F_{\mathit{pinned}}$ (TOFU check)
      \item Namespace in attestation matches $n$
    \end{enumerate}
  \item If any verification fails, reject the artifact and report the failure.
  \item Extract the inner archive. Verify content hash $H_c$ by recomputation. Verify per-file checksums.
\end{enumerate}

\subsection{Three-Layer Cryptographic Defense}

The combination of registry-side namespace enforcement, consumer-side authoritative binding, and artifact-level registry attestation creates three defense layers against dependency confusion (Figure~\ref{fig:defense}):

\begin{description}
  \item[Layer 1: Namespace enforcement at the registry.] Publishing to namespace \texttt{@acme} on the public registry requires verified membership in the \texttt{@acme} organization. An external attacker cannot publish under \texttt{@acme} without organizational membership.

  \item[Layer 2: Authoritative binding at the consumer.] The consumer's project manifest declares \texttt{@acme} as authoritative for \texttt{registry.acme.com}. Even if an attacker publishes \texttt{@acme/utils} on the public registry (bypassing Layer~1), the consumer's resolver rejects it, because \texttt{@acme} is bound to a different registry.

  \item[Layer 3: Registry attestation at the artifact.] The installed artifact carries the distributing registry's Ed25519 signature with fingerprint $F_R$. Even if Layers 1 and 2 are bypassed (registry compromise and consumer misconfiguration), the artifact's attestation fingerprint does not match the consumer's pinned fingerprint for \texttt{@acme}.
\end{description}

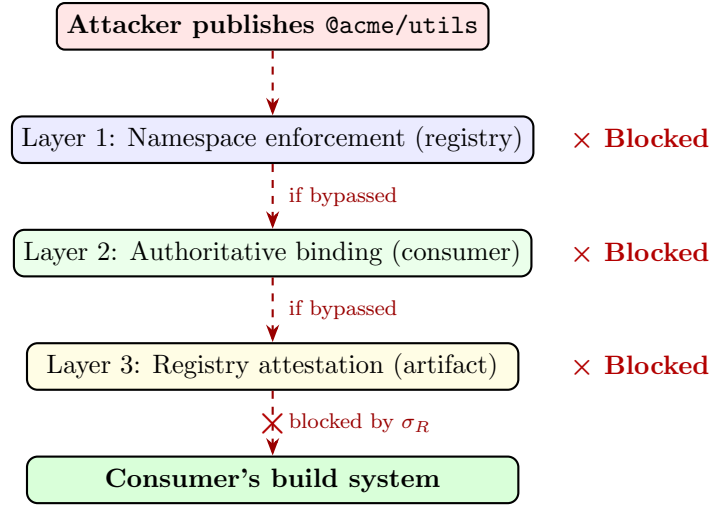
\begin{figure}[t]
\centering
\begin{tikzpicture}[
  >=Stealth,
  layer/.style={draw, thick, rounded corners=4pt, minimum width=6.5cm, minimum height=1.6em, font=\small},
  attacker/.style={draw, thick, fill=red!10, rounded corners=3pt, minimum height=1.6em, font=\small\bfseries},
  blocked/.style={font=\small\bfseries\color{red!70!black}},
  arrow/.style={->, thick},
]

\node[attacker] (atk) at (0, 4.5) {Attacker publishes \texttt{@acme/utils}};

\node[layer, fill=blue!8] (l1) at (0, 3.0) {Layer 1: Namespace enforcement (registry)};
\node[blocked, anchor=west] (b1) at (3.8, 3.0) {\texttimes\ Blocked};

\node[layer, fill=green!8] (l2) at (0, 1.5) {Layer 2: Authoritative binding (consumer)};
\node[blocked, anchor=west] (b2) at (3.8, 1.5) {\texttimes\ Blocked};

\node[layer, fill=yellow!12] (l3) at (0, 0.0) {Layer 3: Registry attestation (artifact)};
\node[blocked, anchor=west] (b3) at (3.8, 0.0) {\texttimes\ Blocked};

\node[draw, thick, fill=green!15, rounded corners=4pt, minimum width=6.5cm, minimum height=1.6em, font=\small\bfseries] (target) at (0, -1.5) {Consumer's build system};

\draw[arrow, red!60!black, dashed] (atk) -- (l1);
\draw[arrow, red!60!black, dashed] (l1) -- node[right, font=\scriptsize, xshift=2pt] {if bypassed} (l2);
\draw[arrow, red!60!black, dashed] (l2) -- node[right, font=\scriptsize, xshift=2pt] {if bypassed} (l3);
\draw[arrow, red!60!black, dashed] (l3) -- node[right, font=\scriptsize, xshift=2pt] {blocked by $\sigma_R$} (target);

\node[font=\Large\bfseries\color{red!70!black}] at (0, -0.75) {\texttimes};

\end{tikzpicture}
\caption{Three-layer cryptographic defense against dependency confusion. Each layer operates independently: namespace enforcement at the registry's API, authoritative binding in the consumer's resolver, and registry attestation at the artifact's cryptographic signature. Under the threat model (Section~\ref{sec:threat-model}), $P(L_3) = 0$ (key compromise excluded), making the attack probability zero regardless of $P(L_1)$ and $P(L_2)$.}
\label{fig:defense}
\end{figure}

\begin{theorem}[Dependency Confusion Resistance]
\label{thm:dep-confusion}
Under the threat model of Section~\ref{sec:threat-model}, a dependency confusion attack against an artifact in an authoritatively bound namespace requires compromise of both the consumer's resolver configuration and the target registry's signing key. Specifically:
\[
P(\text{attack}) \leq P(L_{\text{config}}) \cdot P(L_3)
\]
where $P(L_{\text{config}})$ is the probability of resolver configuration compromise (encompassing both namespace enforcement and authoritative binding, which share a configuration surface), and $P(L_3)$ is the probability of possessing the target registry's Ed25519 private key. In the general case, $P(L_3) = \varepsilon$ for some small $\varepsilon > 0$ reflecting imperfect key protection; under our strict threat model (key compromise excluded), $\varepsilon = 0$ and $P(\text{attack}) = 0$. Under configuration separation (Remark~\ref{rec:config-separation}), the three layers are fully independent and the tighter bound $P(\text{attack}) = P(L_1) \cdot P(L_2) \cdot P(L_3)$ holds.
\end{theorem}

\begin{proof}
By the threat model, the adversary cannot obtain the private registry's Ed25519 key (assumption f). The registry attestation's signature $\sigma_R$ is verified against the registry's public key. An artifact from any other registry (including the adversary's) carries a different registry's signature with a different fingerprint $F_R' \neq F_{\mathit{pinned}}$. TOFU verification (Definition~\ref{def:tofu}) rejects the artifact. Since $P(L_3) = 0$ under our assumptions, $P(L_{\text{config}}) \cdot P(L_3) = 0$ regardless of $P(L_{\text{config}})$.

Layers $L_1$ (namespace enforcement) and $L_2$ (authoritative binding) share a configuration surface: both namespace bindings and TOFU fingerprints reside in the consumer's project configuration file (\texttt{mashin.toml}). An attacker who can modify resolver configuration may weaken both layers simultaneously, so we conservatively bound them as a single factor $P(L_{\text{config}})$. Layer $L_3$ (registry attestation) operates on the artifact's cryptographic signature, independent of consumer configuration. This yields $P(\text{attack}) \leq P(L_{\text{config}}) \cdot P(L_3)$.

Under configuration separation (Remark~\ref{rec:config-separation}), namespace bindings and TOFU fingerprints are stored in separate files, making all three layers independent: $P(\text{attack}) = P(L_1) \cdot P(L_2) \cdot P(L_3)$.
\end{proof}

\begin{remark}[Configuration Separation]
\label{rec:config-separation}
To achieve full layer independence, we recommend storing TOFU fingerprints in a separate integrity-protected file (analogous to SSH's \texttt{known\_hosts}), distinct from the project's resolver configuration. This ensures that an attacker who modifies namespace bindings in \texttt{mashin.toml} cannot simultaneously subvert fingerprint verification. Under this separation, the three defense layers are genuinely independent, and the multiplicative probability bound holds.
\end{remark}

\section{Six-Level Verification Chain}
\label{sec:verification-chain}

The archive format, signatures, and namespace binding compose into a six-level verification chain where each level builds on the previous:

\begin{enumerate}
  \item \textbf{File Integrity.} Each source file's SHA-256 matches per-file checksums in the provenance manifest. \emph{Proves:} no individual file was corrupted or substituted.

  \item \textbf{Artifact Identity.} Content hash $H_c$ recomputed from extracted files matches the provenance manifest. \emph{Proves:} the artifact is a specific, deterministic snapshot.

  \item \textbf{Publisher Authenticity.} Ed25519 signature $\sigma_P$ verifies against the publisher's public key. \emph{Proves:} the claimed publisher actually signed this provenance manifest.

  \item \textbf{Envelope Integrity.} CHECKSUM matches SHA-256 of \texttt{contents.tar.gz}. \emph{Proves:} the compressed archive was not modified after packaging.

  \item \textbf{Registry Attestation.} Ed25519 signature $\sigma_R$ verifies against the registry's public key. Fingerprint $F_R$ matches pinned fingerprint (if configured). \emph{Proves:} a specific, identified registry distributed this artifact.

  \item \textbf{Lineage Provenance} (optional). The provenance manifest's evolution anchor references a specific event in an external provenance ledger. The event hash is verified against the ledger's hash chain. \emph{Proves:} the artifact's complete development history (creation, edits, tests, promotions) is tamper-evident and auditable. This level requires an external provenance system (see Section~\ref{sec:integration} for a case study).
\end{enumerate}

\paragraph{Verification modes.} Levels 1--3 are verified on every install (fast, offline-capable). Levels 4--5 are verified in \emph{strict} mode. Level 6 is optional and verified in \emph{full} mode (requires network access to query an external provenance ledger). The core contribution of this paper is Levels 1--5; Level 6 demonstrates extensibility to systems with development-phase provenance.

\begin{proposition}[Verification Composability]
Each verification level $V_i$ produces a boolean result. The composite verification result is $V = \bigwedge_{i=1}^{n} V_i$ where $n \in \{3, 5, 6\}$ depends on the verification mode. Each level's verification is independent: failure at level $i$ does not prevent verification at level $j \neq i$. This enables precise error reporting (``publisher signature valid but registry attestation fingerprint mismatch'').
\end{proposition}

\section{Formal Properties}
\label{sec:formal}

We formalize the security properties of the distribution provenance system.

\begin{definition}[Distribution Provenance Record]
\label{def:provenance-record}
A distribution provenance record for artifact $a$ is a tuple:
\[
D(a) = (B, \sigma_P, F_P, A, \sigma_R, F_R, H_{\text{check}})
\]
where $B$ is the provenance manifest, $\sigma_P$ is the publisher signature, $F_P$ is the publisher fingerprint, $A$ is the registry attestation, $\sigma_R$ is the registry signature, $F_R$ is the registry fingerprint, and $H_{\text{check}}$ is the inner archive checksum.
\end{definition}

\begin{theorem}[Non-Repudiation]
\label{thm:non-repudiation}
Given a valid distribution provenance record $D(a)$:
\begin{enumerate}
  \item The publisher cannot deny having signed the provenance manifest: $\text{Ed25519-Verify}(k_P^{\text{pub}}, \text{canonical}(B), \sigma_P) = \text{true}$ uniquely identifies the holder of $k_P^{\text{priv}}$.
  \item The registry cannot deny having distributed the artifact: $\text{Ed25519-Verify}(k_R^{\text{pub}}, \text{canonical}(A), \sigma_R) = \text{true}$ uniquely identifies the holder of $k_R^{\text{priv}}$.
  \item Neither party can deny the temporal ordering: $t_P < t_R$ where $t_P$ is the publisher signing timestamp and $t_R$ is the attestation acceptance timestamp.
\end{enumerate}
\end{theorem}

\begin{proof}
Ed25519 is existentially unforgeable under chosen-message attacks (EU-CMA) assuming the hardness of the discrete logarithm problem on Curve25519~\cite{bernstein2012ed25519}. Under this assumption, $\sigma_P$ could only have been produced by the holder of $k_P^{\text{priv}}$, and $\sigma_R$ could only have been produced by the holder of $k_R^{\text{priv}}$. The temporal ordering follows from the countersigning structure: the registry generates $A$ containing $\text{SHA-256}(\text{canonical}(B))$ and $t_R$, and $B$ contains $t_P$ via the publisher's signing timestamp. The registry cannot generate $A$ before receiving $B$ (it hashes $B$), establishing $t_P < t_R$.
\end{proof}

\begin{theorem}[Tamper Evidence]
\label{thm:tamper-evidence}
Any modification to the artifact's contents after publication is detectable:
\begin{enumerate}
  \item Modifying any source file invalidates its per-file checksum (Level~1).
  \item Modifying any source file invalidates $H_c$ (Level~2).
  \item Modifying the provenance manifest invalidates $\sigma_P$ (Level~3).
  \item Modifying \texttt{contents.tar.gz} invalidates CHECKSUM (Level~4).
  \item Modifying any outer envelope document invalidates $\sigma_R$ (since $A$ contains $\text{SHA-256}(B)$) (Level~5).
\end{enumerate}
\end{theorem}

\begin{proof}
SHA-256 is collision-resistant: finding $x' \neq x$ such that $\text{SHA-256}(x') = \text{SHA-256}(x)$ requires $O(2^{128})$ operations. Ed25519 signatures are unforgeable: producing a valid signature without the private key requires solving the discrete logarithm on Curve25519. Modification of any component at any level invalidates at least one hash or signature verification.
\end{proof}

\begin{definition}[Lineage Integrity]
\label{def:lineage}
The version chain of an artifact is a sequence $a_1, a_2, \ldots, a_n$ where each $a_i$'s provenance manifest contains $\textit{parent\_version} = v_{i-1}$ and $\textit{parent\_content\_hash} = H_c(a_{i-1})$. The chain has integrity if:
\[
\forall i \in [2, n] : B(a_i).\textit{parent\_content\_hash} = H_c(a_{i-1})
\]
Any modification to a historical version $a_j$ (with $j < n$) is detectable because $H_c(a_j)$ changes, breaking the chain at $a_{j+1}$.
\end{definition}

\section{Ecosystem Comparison}
\label{sec:comparison}

Table~\ref{tab:comparison} compares the distribution provenance capabilities of the present system against eight major package ecosystems across nine dimensions.

\begin{table*}[t]
\centering
\caption{Distribution provenance capabilities across package ecosystems. \textbf{Bold} indicates the strongest capability in each column. The final four columns (Trust Enforce through Gov.\ Policy) represent governance-enforced resolution capabilities that are outside the traditional scope of package distribution; they are included to show the design space extension, not as a like-for-like comparison.}
\label{tab:comparison}
\resizebox{\textwidth}{!}{%
\footnotesize
\setlength{\tabcolsep}{6pt}
\begin{tabular}{@{}l l l l l l l l l l@{}}
\toprule
\textbf{Ecosystem} & \textbf{Pkg Signing} & \textbf{Reg.\ Identity} & \textbf{Reg.\ Countersig} & \textbf{Consumer Enforce} & \textbf{Dep.\ Confusion} & \textbf{Trust Enforce} & \textbf{Effect Escalation} & \textbf{Comp.\ Level} & \textbf{Gov.\ Policy} \\
\midrule
npm & Sigstore (opt-in) & None & None & None & \texttt{.npmrc} config & None & None & None & None \\
Cargo/crates.io & \texttt{cargo-vet} (3P) & None$^\dagger$ & None & None & N/A (singleton) & None & None & None & None \\
Hex.pm & Protobuf sig & Partial$^\ddagger$ & Signed repo name & None & Repo name in signed data & None & None & None & None \\
PyPI/pip & PEP 740 (Sigstore) & None & None & None & \texttt{--index-url} flag & None & None & None & None \\
Go modules & \texttt{go.sum} via sumdb & Transparency log & None & \texttt{GOPROXY} config & Checksum DB & None & None & None & None \\
Docker/OCI & Notary/cosign & URL in image ref & None & None & Hostname in ref & None & None & None & None \\
NuGet & X.509 author sig & Repo countersig & X.509 (optional) & Not by default & Optional repo sig & None & None & None & None \\
Maven & PGP (optional) & None & None & None & \texttt{<repositories>} XML & None & None & None & None \\
\midrule
\textbf{Present work} & \textbf{Ed25519 (req.)} & \textbf{Ed25519/instance} & \textbf{Mandatory} & \textbf{Pinned + auth.} & \textbf{3-layer cryptographic} & \textbf{Min.\ tier req.} & \textbf{Transitive} & \textbf{Structural} & \textbf{Org.\ policies} \\
\bottomrule
\end{tabular}%
}

\vspace{0.3em}
{\scriptsize $^\dagger$Crates.io is a singleton registry; no private registries exist in the standard toolchain. $^\ddagger$Hex.pm signs a \texttt{repository} field inside the package protobuf, identifying the source repo by name.}
\end{table*}

\paragraph{Key findings.} No existing ecosystem combines all four capabilities: mandatory publisher signing, cryptographic registry identity, mandatory registry countersigning, and consumer-side cryptographic enforcement with namespace binding. NuGet is the closest, providing optional X.509 repository countersignatures, but these are not enforced during installation. Hex.pm signs a repository name inside the package metadata (the closest data-level prior art) but provides no consumer-side enforcement mechanism. Beyond supply-side provenance, the comparison reveals a second differentiator: no existing ecosystem performs governance-enforced dependency resolution. All eight ecosystems resolve dependencies based solely on version compatibility. Extending resolvers to enforce trust tier requirements and detect effect permission escalation across transitive dependencies is discussed in Section~\ref{sec:discussion}.

The Go module ecosystem's \texttt{sum.golang.org} transparency log is architecturally distinct: it provides append-only checksum verification (preventing modification of published modules) but does not address registry-level provenance (which registry served the module) or dependency confusion (the checksum DB records all modules, including attacker-published ones).

\section{Case Study: Integration with Runtime Governance}
\label{sec:integration}

The distribution provenance system described in Sections~\ref{sec:archive-format}--\ref{sec:verification-chain} is ecosystem-agnostic: any package registry can implement cryptographic registry identity, dual signatures, and authoritative namespace binding. To demonstrate integration with a broader governance architecture, we describe an implementation within Mashin, a platform for AI workflows that implements a three-layer structural governance architecture~\cite{mccann2026structural}.

Mashin enforces governance by construction at three layers: the definition layer (an append-only Evolution Provenance Ledger recording all artifact changes as hash-chained events), the capability layer (effect isolation with sandboxed execution), and the execution layer (a directive-based pure execution model where executors return declarative effect directives processed by a governed interpreter). Distribution provenance extends this to a four-phase lifecycle chain:

\begin{enumerate}
  \item \textbf{Definition} (Evolution Provenance Ledger): hash-chained events recording creation, edits, tests, and promotions. Each artifact version is assigned a content hash $H_c$.
  \item \textbf{Distribution} (this work): the provenance manifest's evolution anchor links the published artifact to its development history. The publisher signature and registry attestation prove authorship and distribution path.
  \item \textbf{Capability} (effect isolation): the provenance manifest's module-level metadata enables permission policy evaluation before compilation.
  \item \textbf{Execution} (directive interpreter): each execution run's provenance record includes the artifact's content hash, the same $H_c$ in the provenance manifest, the same $H_c$ in the definition ledger.
\end{enumerate}

\begin{proposition}[Continuous Provenance Chain]
The four phases are linked by content hashes: the definition ledger assigns $H_c$ to each artifact version; the provenance manifest records $H_c$ and the evolution anchor; each execution run records $H_c$ in its provenance output. An auditor can traverse the chain in either direction: from execution output back through distribution to definition, or forward from creation through publication to execution.
\end{proposition}

This case study illustrates a general principle: distribution provenance is most valuable when integrated into a broader governance chain. Any system with build provenance, deployment attestation, or runtime auditing can extend its governance to the distribution phase using the mechanisms described in this paper. The integration is particularly relevant for domains where governance and auditability are regulatory requirements (AI workflows, financial services, healthcare), where the governance chain must be unbroken even as artifacts cross organizational boundaries through registries.

\section{Evaluation}
\label{sec:evaluation}

\subsection{Verification Performance}

All performance numbers in this section are measured from a running implementation on an Apple M-series processor, using median values over 50 iterations. The implementation uses Erlang's \texttt{:crypto} module for Ed25519 and SHA-256 operations.

Table~\ref{tab:perf} shows end-to-end verification latency for a 4-machine Krate across the three verification modes.

\begin{table}[htbp]
\centering
\caption{Measured Krate verification latency (4 machines, $n=50$, 5-iteration warmup).}
\label{tab:perf}
\small
\begin{tabular}{@{}llrrr@{}}
\toprule
\textbf{Mode} & \textbf{Levels} & \textbf{Median} & \textbf{Mean} & \textbf{p99} \\
\midrule
Default & 1--3 (file, identity, publisher) & 1.18~ms & 1.23~ms & 2.15~ms \\
Strict & 1--5 (+ envelope, registry) & 1.61~ms & 2.54~ms & 12.41~ms \\
Full & 1--6 (+ lineage provenance) & 1.13~ms & 1.25~ms & 4.12~ms \\
\bottomrule
\end{tabular}
\end{table}

Verification completes in low single-digit milliseconds for all modes. The strict mode shows higher variance (p99 of 12.41~ms) due to the additional envelope and registry attestation checks. Full mode is comparable to default mode because lineage verification in this measurement uses a local provenance ledger rather than a network round-trip. In production deployments with remote ledgers, Level~6 would add network latency.

\paragraph{Comparison context.} For reference, \texttt{npm audit signatures} performs Sigstore signature verification for each installed package, adding roughly 100--500~ms per package depending on network latency to the Rekor transparency log. Cargo's \texttt{crates.io} verifies SHA-256 checksums against a locally-cached index (sub-millisecond, comparable to our Level~1). Our system's verification adds two Ed25519 signature checks (publisher and registry) beyond checksum verification. The total cost of 1--2~ms is comparable to or faster than existing ecosystem verification because it avoids network round-trips to transparency logs for Levels 1--5.

Table~\ref{tab:primitives} breaks down the cost of individual cryptographic operations.

\begin{table}[htbp]
\centering
\caption{Individual operation latency (median of 50 iterations).}
\label{tab:primitives}
\small
\begin{tabular}{@{}lr@{}}
\toprule
\textbf{Operation} & \textbf{Latency} \\
\midrule
Ed25519 sign & 33~\textmu s \\
Ed25519 verify & 41~\textmu s \\
SHA-256 (1~KB) & $<$1~\textmu s \\
SHA-256 (10~KB) & 4~\textmu s \\
SHA-256 (100~KB) & 31~\textmu s \\
Canonical encode (small) & $<$1~\textmu s \\
Canonical encode (large, 20 machines) & 8~\textmu s \\
Tar extract (outer envelope) & 259~\textmu s \\
\bottomrule
\end{tabular}
\end{table}

The dominant cost in verification is tar extraction of the outer envelope (259~\textmu s), not cryptographic computation. Ed25519 verification (41~\textmu s) and SHA-256 hashing are negligible even for large artifacts.

The full end-to-end lifecycle (build + sign + attest + verify) for a 4-machine project completes in \textbf{4.51~ms median} (5.55~ms mean). Ed25519 signing and SHA-256 hashing both complete in single-digit microseconds.

\subsection{Archive Format Overhead}

The two-layer format adds the provenance envelope to each artifact (Table~\ref{tab:envelope}):

\begin{table}[htbp]
\centering
\caption{Archive overhead for projects of varying size.}
\label{tab:envelope}
\small
\begin{tabular}{@{}lrrr@{}}
\toprule
\textbf{Project} & \textbf{Source Size} & \textbf{Krate Size} & \textbf{Overhead} \\
\midrule
1 machine & 203~B & 5.6~KB & 2,659\% \\
4 machines & 812~B & 6.7~KB & 725\% \\
12 machines & 2.4~KB & 10.2~KB & 308\% \\
\bottomrule
\end{tabular}
\end{table}

The overhead percentage is high for minimal projects (the provenance envelope and signature documents have a fixed minimum size of several kilobytes), but decreases rapidly as source size grows: from 2,659\% for a single-machine project down to 308\% for a 12-machine project. For production projects with meaningful source content, the absolute envelope size (several KB) becomes negligible relative to the machine source.

\subsection{Security Analysis}

We evaluate the system against the threat model of Section~\ref{sec:threat-model} (Table~\ref{tab:security}):

\begin{table}[htbp]
\centering
\caption{Attack resistance by defense layer.}
\label{tab:security}
\small
\begin{tabular}{@{}lcccc@{}}
\toprule
\textbf{Attack} & \textbf{L1} & \textbf{L2} & \textbf{L3} & \textbf{Result} \\
\midrule
Dependency confusion & \checkmark & \checkmark & \checkmark & Blocked \\
Namespace squatting & \checkmark & -- & \checkmark & Blocked \\
Tampered artifact & -- & -- & \checkmark$^*$ & Detected \\
Registry DB compromise & -- & -- & \checkmark$^*$ & Detected \\
Typosquatting & -- & -- & -- & Not addressed$^\dagger$ \\
\bottomrule
\end{tabular}

\vspace{0.3em}
{\footnotesize L1 = namespace enforcement, L2 = authoritative binding, L3 = registry attestation. \checkmark\ = defended. $^*$Publisher signature provides tamper detection independent of registry; ``Registry DB compromise'' assumes database-level access without key compromise (the signing key is assumed protected per assumption~f of the threat model). $^\dagger$Typosquatting requires name similarity analysis, orthogonal to cryptographic provenance.}
\end{table}

\paragraph{Residual risks.} The system does not address: (1)~typosquatting (publishing \texttt{@community/sIack} to impersonate \texttt{@community/slack}), which requires name similarity analysis at the registry; (2)~insider threats with registry server access (key compromise), which requires HSMs and operational security; (3)~social engineering of authorized publishers (the xz-utils vector), which requires behavioral analysis beyond cryptographic verification.

\section{Discussion}
\label{sec:discussion}

\subsection{Limitations}

\paragraph{Key management and revocation.} Developers must manage Ed25519 keypairs and each registry instance must manage its own keypair. Auto-generation on first use and secure local storage mitigate the developer burden, but key rotation, backup, and revocation remain operational concerns. The current design has no key revocation mechanism: if a publisher's private key is compromised, there is no protocol for invalidating artifacts signed with that key or for propagating revocation to consumers who have already installed those artifacts. The TUF framework~\cite{samuel2010survivable} addresses this with threshold signatures and delegated trust, but integrating TUF-style key management adds significant complexity. For enterprise deployments, hardware security modules (HSMs) or secret management systems (Vault) are recommended for registry keys. A revocation protocol is identified as future work.

\paragraph{TOFU trust establishment.} The first time a consumer configures a private registry, they must trust the fingerprint they receive. If the initial connection is compromised (man-in-the-middle), the consumer pins the attacker's fingerprint. This is the same limitation as SSH's \texttt{known\_hosts}, well-understood and acceptable for most deployment scenarios, but not suitable for zero-trust environments without an out-of-band fingerprint verification channel.

\paragraph{Configuration required for full defense.} The three-layer defense requires consumers to configure registry bindings with authoritative namespace binding. Without explicit configuration, the resolver uses the default public registry and Layers 1--2 provide no protection. Layer 3 (registry attestation) still provides post-installation auditability (``which registry did this come from?'') but not prevention. The full defense is opt-in per namespace, which is appropriate for enterprise scenarios where private namespaces exist.

\paragraph{Scalability of registry identity.} Each registry instance has its own keypair. In a large organization with many team-level registries, key management complexity grows linearly. The trust hierarchy (parent registry) helps organize relationships but does not reduce the number of keys. Federation protocols for cross-registry trust have not yet been specified.

\subsection{Future Work}

\paragraph{Transparency log.} When the public registry exceeds 1000 published artifacts, integrating a transparency log (similar to \texttt{sum.golang.org} or Sigstore's Rekor) would provide an independent, append-only record of all publication events. This complements the dual-signature model: the signatures prove publisher and registry identity; the transparency log proves the publication was publicly witnessed.

\paragraph{Formal verification of the verification chain.} The six-level verification chain's security properties (Theorems~\ref{thm:dep-confusion}--\ref{thm:tamper-evidence}) could be mechanized in Rocq or Lean, providing machine-checked proofs rather than paper proofs.

\paragraph{Key rotation under Byzantine assumptions.} The current key rotation protocol assumes an honest registry that publishes accurate key validity periods. Under Byzantine assumptions (a compromised registry publishing false key histories), additional mechanisms, such as key commitments in a transparency log, would be needed.

\paragraph{AI-generation provenance.} The provenance manifest's extensible lineage section could carry structured metadata about AI involvement in artifact creation: which model was used, at what composition level (assisted, directed, or autonomous), and whether the output was human-reviewed. Because the lineage is covered by both publisher and registry signatures, this metadata would be cryptographically attested rather than self-reported. No existing package ecosystem tracks AI-generation provenance in any form. Consumer-side governance policies that filter dependencies by AI composition level are a natural extension.

\paragraph{Governance-enforced dependency resolution.} Current package resolvers check only version compatibility. In a system where every artifact carries governance metadata (trust tiers, effect declarations, composition levels), the resolver could enforce governance constraints during installation rather than relying on post-install audit tools. The most interesting check is transitive effect escalation detection: flagging when a new dependency transitively introduces I/O capabilities not present in the current dependency tree. The governance surface of a dependency set is monotone (adding dependencies can only increase exposure), which makes incremental reporting feasible.

\paragraph{Cross-registry federation and federated governance.} Specifying protocols for cross-registry trust (including cross-signing, trust delegation, and federated namespace resolution) would enable multi-organization deployment scenarios where artifacts flow between registries with different trust domains. Extending governance-enforced resolution to federated scenarios raises open questions: when two organizations with different governance policies share dependencies, how should the resolver reconcile conflicting trust tier requirements or effect policies? A governance negotiation protocol for cross-organization resolution is a natural extension of the present work.

\section{Conclusion}
\label{sec:conclusion}

We have presented a cryptographic artifact distribution provenance system that eliminates dependency confusion attacks through structural design. The system comprises three components (cryptographic registry identity, a dual-signature distribution model, and authoritative namespace binding) that create three independent cryptographic barriers against supply chain attacks. A systematic comparison across eight major package ecosystems confirms that no existing system combines all four capabilities: mandatory publisher signing, cryptographic registry identity, mandatory registry countersigning, and consumer-side cryptographic enforcement.

The key insight is that dependency confusion is not a configuration problem; it is a provenance problem. The artifact itself must carry cryptographic proof of its distribution path. Configuration-based defenses (registry URL settings, index URL flags, proxy configurations) operate outside the artifact and fail silently when misconfigured. Artifact-level cryptographic provenance is verifiable regardless of configuration state.

As demonstrated through a case study (Section~\ref{sec:integration}), the distribution provenance system integrates naturally with broader governance architectures. When combined with definition-phase provenance and execution-phase auditing, it extends to a four-phase lifecycle chain covering definition, distribution, capability, and execution. Content hashes link all four phases into a continuous, tamper-evident chain, enabling end-to-end verification from an artifact's creation through publication and distribution to execution.

\bibliographystyle{plainnat}
\bibliography{distribution-references}

\end{document}